\documentclass[11pt,a4paper]{article}
\usepackage{tikz,amssymb,amsmath,amsthm,multirow,color,xcolor,cancel,bbm}
\usepackage[textheight=25cm,textwidth=16cm]{geometry}
\newtheorem{proposition}{Proposition}
\newtheorem{lemma}{Lemma}

\newtheorem{question}{Question}
\newtheorem{remark}{Remark}

\newtheorem{theorem}{Theorem}
\newtheorem{definition}{Definition}
\newcounter{claimcounter}
\newtheorem*{claim*}{{\it Claim}}
\newtheorem*{claim}{\stepcounter{claimcounter}{\it Claim \theclaimcounter}}
\newenvironment{subproof}{\begin{proof}[Subproof]}{\end{proof}}
\newcommand{\N}{\mathbbm{N}}
\newcommand{\cst}{\mathrm{cst}}
\newcommand{\One}[1]{{\mathbbm 1}{\{#1\}}}

\title{Simple dynamics on graphs}

\author{
Maximilien Gadouleau\footnote{School of Engineering and Computing Sciences, Durham University, UK. \tt{m.r.gadouleau@durham.ac.uk}}
\and 
Adrien Richard\footnote{Laboratoire I3S, CNRS \& Universit\'e de Nice-Sophia Antipolis, France.
\tt{richard@unice.fr}}
\footnote{Corresponding author.}  
}

\date{December 8, 2015\footnote{This work is partially supported by CNRS and Royal Society through the International Exchanges Scheme grant {\em Boolean networks, network coding and memoryless computation.}}}

\begin{document}

\maketitle

\begin{abstract}
Does the interaction graph of a finite dynamical system can force this system to have a ``complex'' dynamics ? In other words, given a finite interval of integers $A$, which are the signed digraphs $G$ such that every finite dynamical system $f:A^n\to A^n$ with $G$ as interaction graph has a ``complex'' dynamics ? If $|A|\geq 3$ we prove that no such signed digraph exists. More precisely, we prove that for every signed digraph $G$ there exists a system $f:A^n\to A^n$ with $G$ as interaction graph that converges toward a unique fixed point in at most $\lfloor\log_2 n\rfloor+2$ steps. The boolean case $|A|=2$ is more difficult, and we provide partial answers instead. We exhibit large classes of unsigned digraphs which admit boolean dynamical systems which converge toward a unique fixed point in polynomial, linear or constant time. 
\end{abstract}


\section{Introduction} \label{sec:intro}

Let $A=\{0,1,\dots,s\}$ be a finite integer interval, and let $n$ be a positive integer. A {\em finite dynamical system} is a function 
\[
f:A^n\to A^n,\qquad  x=(x_1,\dots,x_n)\mapsto f(x)=(f_1(x),\dots,f_n(x)).
\]
If $|A|=2$, such a system is called {\em boolean network}. Finite dynamical systems, and boolean networks in particular, have many applications: they have been used to model gene networks \cite{K69,T73,TK01,J02}, neural networks \cite{MP43,H82,G85,GM91}, social interactions \cite{PS83, GT83} and more (see \cite{TA90,GM90}). 

\medskip
The structure of a finite dynamical system $f$ can be represented via its {\em interaction graph} $G$, which roughly describes the dependencies between the variables of the systems (depending on the context, this graph is sometimes called {\em dependency graph}, {\em influence graph} or {\em regulatory graph}). More formally, $G$ is a digraph with vertex set $\{1,\dots,n\}$ and an arc from $j$ to $i$ if $f_i(x)$ depends on $x_j$. An arc from $j$ to $i$ can also be labeled by a sign indicating whether $f_i(x)$ is an increasing (positive sign), decreasing (negative sign), or non-monotone (zero sign) function of $x_j$. This is commonly the case when modelling gene networks, since a gene can typically either activate (positive sign) or inhibit (negative sign) another gene. 

\medskip
In many contexts, as in molecular biology, the interaction graph is known--or at least well approximated--, while the actual function $f$ is not. A natural and difficult question is then the following: {\em what can be said on system $f:A^n\to A^n$ according to its interaction graph only?} Among the many dynamical properties that can be studied, fixed points are crucial because they represent stable states \cite{R86,TA90,GM91}. As such, they are arguably the property which has been the most thoroughly studied (see \cite{R86,SD05,RRT08,A08,GS10,GRS14} and the references therein).

\medskip
In this paper, we are interested in ``simple'' dynamics, considering that a dynamics is simple if it describes a fast convergence toward a unique fixed point. Formally, $f$ converges towards a unique fixed point in $k$ steps if $f^k$ is a constant. In that case, we say that $f$ is a {\em nilpotent function} and the minimal $k$ such that $f^k$ is a constant is called the {\em class} of $f$. Also, we say that a signed or unsigned digraph $G$ {\em admits} a function $f$ if $G$ is the signed or unsigned version of the interaction graph of $f$. 

\medskip
A fundamental result of Robert is the following: {\em if the interaction graph of $f:A^n\to A^n$ is acyclic then $f$ is a nilpotent function of class at most $n$} \cite{R86}. This shows that ``simple'' interaction graphs imply ``simple'' dynamics. But conversely, does ``complex'' interaction graphs imply ``complex'' dynamics ? More precisely, which are the interaction graphs that can force a system to have a non simple dynamics? This is the question we study in this paper. 

\medskip
We first study the non-boolean case $|A|\geq 3$ in Section~\ref{sec:non-boolean_case}. Essentially, we show that every signed digraph $G$ on $n$ vertices admits a nilpotent function $f:A^n\to A^n$ of class at most $\lfloor\log_2 n\rfloor+2$. Furthermore, if $|A|>3$ then the upper-bound on the class of $f$ can be reduced to only $2$. Hence, in the non-boolean case, we cannot conclude that a system $f$ has a non simple dynamics from its interaction graph only.     

\medskip
We then study the boolean case $|A|=2$ in Section~\ref{sec:boolean_case}, which is more difficult. First, not all digraphs admit a boolean nilpotent function. The directed cycle is the most simple example, and it seems very difficult to characterize the digraphs that admit a boolean nilpotent function. Thus we provide partial answers. We exhibit large classes of unsigned digraphs which admit boolean dynamical systems which converge toward a unique fixed point in polynomial, linear or constant time. In particular, we prove that if $G$ has a primitive spanning strict subgraph then $G$ admits a boolean nilpotent function of class at most $n^2-2n+3$. We also prove that if $G$ is strongly connected and if the out-neighborhood of some vertex of $G$ induces a non-acyclic digraph, then $G$ admits a boolean nilpotent function $f$ of class at most $2n-1$. Besides, we prove that if $G$ is a loop-less connected symmetric digraph with at least three vertices, then $G$ admits a boolean nilpotent function $f$ of class $3$. We have not been able to prove or disprove the following assertion: there exists a constant $c$ such that, for every digraph $G$ with $n$ vertices, if $G$ admits a boolean nilpotent function, then $G$ admits a boolean nilpotent function of class at most $cn$.

\section{Preliminaries} \label{sec:def}

The vertex set of a digraph $G$ is denoted $V(G)$ and its arc set, which is a subset of $V(G)\times V(G)$, is denoted $A(G)$. The in-neighborhood of a vertex $v$ is denoted $G(v)$; this is an non-usual but very convenient notation for our purpose. Other notations and terminologies on digraphs are usual and consistent with \cite{BG08}. Paths and cycles of are always directed, without repetition of vertices, and seen as subgraphs. The subgraph of $G$ induced by a set of vertices $I\subseteq V(G)$ is denoted $G[I]$. If $X$ is an arc, a vertex, a set of arcs, or a set of vertices, then $G\setminus X$ is the subgraph obtain from $G$ by removing $X$ or the elements in $X$. We say that $G$ is {\em strong} if $G$ is strongly connected. A strongly connected component $I$ ({\em strong component} for short) of $G$ is {\em initial} if there is no arc $(u,v)$ with $u\not\in I$ and $v\in I$. If $G$ and $G'$ are two digraphs, then $G\cup G'$ is the digraph with vertex set $V(G)\cup V(G')$ and arc set $A(G)\cup A(G')$. A digraph {\em on} a set $V$ is a digraph with vertex set $V$. A {\em tree} is a digraph in which all the vertices have in-degree one, excepted one vertex, called the {\em root}, which has in-degree zero. A {\em forest} is a digraph in which all the connected components are trees.  A {\em loop} is an arc from a vertex to itself. A vertex is {\em linear} if it has a unique in-neighbor and a unique out-neighbor.

\medskip
A signed digraph $G$ consists in a digraph, denoted $|G|$, together with a map that labels each arc of $|G|$ by a positive, negative or null sign. We say that an arc is {\em signed} if it is positive or negative, and {\em unsigned} otherwise. The digraph obtained from $G$ by keeping only positive arcs is denoted $G^+$. We define similarly $G^-$ and $G^0$. The digraph obtained by keeping only signed arcs is denoted $G^\pm$ (thus $G^\pm=G^+\cup G^-$). A cycle of $G$ is positive (resp. negative) if it contains an unsigned arc or an even (resp. odd) number of negative arcs. In the following, all graph-theoretic concepts that do not involve signs are applied on $G$ or $|G|$ indifferently.

\medskip
Let $A$ be a finite interval of integers, let $n$ be a positive integer and $[n]=\{1,\dots,n\}$. A function {\em over} $A$ is a map $f:A^n\to A^n$. A function over $\{0,1\}$ is a {\em boolean function}. As usual, for all $k\in\mathbb{N}$ we set $f^k=\mathrm{id}$ if $k=0$ and $f^{k}=f\circ f^{k-1}$ otherwise. If $f$ is any function, we write $f=\cst$ to mean that $f$ is a constant. In the following, functions are often defined using conjunctions ($\land$) disjunctions ($\lor$) and exclusive disjunctions ($\oplus$). If $I\subseteq [n]$ and $x\in\{0,1\}^I$ then, by convention, $\lor_{i\in I} x_i=\oplus_{i\in I}x_i=0$ and $\land_{i\in I}x_i=1$ if $I$ is empty, and $\lor_{i\in I} x_i=\oplus_{i\in I}x_i=\land_{i\in I}x_i=x_i$ if $I=\{i\}$. 

\begin{definition}
A function $f$ over $A$ is {\em nilpotent} if there exists $k\in\mathbb{N}$ such that $f^k$ is constant. If $f$ is nilpotent, then the smallest $k$ such that $f^k$ is a constant is the {\em class} of $f$. 
\end{definition}

\begin{definition}
The {\em interaction graph} of a function $f$ over $A$ is the signed digraph $G(f)$ on $[n]$ with arcs defined as follows: for all $j,i\in [n]$, there is an arc $(j,i)$ if $f_i(a)\neq f_i(b)$ for some $a,b\in A^n$ such that $a_j < b_j$ and $a_k = b_k$ for all $k \neq j$; and an arc $(j,i)$ is positive if $f_i(a) \leq f_i(b)$ for all such $a$ and $b$, negative if $f_i(a) \geq f_i(b)$ for all such $a$ and $b$, and null otherwise.
\end{definition}

Hence, $G(f)$ has an arc $(j,i)$ if and only if $f_i(x)$ depends essentially on $x_j$, and the sign of an arc $(i,j)$ is positive (resp. negative) if an only if for every fixed $x_k$, $k\neq j$, $f_i(x)$ is a non-decreasing (resp. non-increasing) function of $x_j$.  

\begin{definition}
Let $G$ be a signed digraph and let $f$ be a function over $A$. If $G(f)=G$ then we say that $G$ {\em admits} $f$ and that $f$ is a {\em $G$-function}. Similarly, if $|G(f)|=|G|$ then we say that $|G|$ {\em admits} $f$ and that $f$ is a {\em $|G|$-function}.
\end{definition}

Given a signed or unsigned digraph $G$, we are interested in the existence of a nilpotent $G$-function. According to the following proposition, it makes sense to focus on the minimal alphabet size for which such a function exists. 

\begin{proposition}\label{pro:alphabet}
Let $A$ and $B$ be two finite intervals of integers with $A\subseteq B$. If a signed digraph admits a nilpotent function of class $k$ over $A$ then it admits a nilpotent function of class~$k$~over~$B$.
\end{proposition}
\begin{proof}
Let $G$ be a signed graph and let $f$ be a nilpotent $G$-function of class $k$ over $A$. For all $x \in B^n$, let $\tilde x$ be the point of $A^n$ that minimizes the Manhattan distance $d(x,\tilde x)=\sum_{i}|x_i-\tilde x_i|$. Let $\tilde f: B^n \to B^n$ be defined by $\tilde f(x) = f(\tilde x)$ for all $x \in B^n$. Then the following three properties are easily checked: $\tilde f$ is a $G$-function; if $f^{k+1} = f^k$ then $\tilde f^{k+1} = \tilde f^k$; and if $f^k$ is a constant, then so is $\tilde f^k$.
\end{proof}

Besides, it is easy to see that every signed digraph $G$ admits a function $f$ over $\{0,1,2\}$. However, some signed digraphs admit no boolean functions (this is a first qualitative difference between boolean and non-boolean alphabets). They are characterized below.   

\begin{proposition}[\cite{PR10}]\label{pro:existence}
A signed digraph $G$ admits a boolean function if and only if $|G^+(i)|+|G^-(i)|\geq 2$ for every vertex $i$ such that $|G^0(i)|=1$. 
\end{proposition}

The following proposition shows that a signed digraph $G$ admits a non-boolean nilpotent function if and only if all the initial strong components of $G$ do. In the boolean case, this is no longer true: additional hypotheses on signs are needed. 

\begin{proposition}\label{pro:strong}
Let $A$ be an integer interval and let $G$ be a signed digraph. If $G$ admits a nilpotent function over the alphabet $A$ then all its initial strong components do, and the converse is true if $|A|\geq 3$ or if $|A|=2$ and all the unsigned arcs of $G$ are inside the initial strong components.
\end{proposition}
\begin{proof}
Let $f$ be a nilpotent $G$-function over $A$, and suppose that $I$ is an initial strong component of $G$. It is easy to see that if $f$ is a nilpotent $G$-function over $A$, then the ``restriction'' of $f$ to $I$, {\em i.e.} the function $\tilde f:A^I\to A^I$ defined by $\tilde f(x_I)=f(x)_I$ for all $x\in A^n$, is a nilpotent $G[I]$-function. This proves the first assertion. 

\medskip
For the converse, suppose that $A=\{0,1,\dots,s\}$, and let $I$ be the set of vertices that belong to an initial strong component of $G$ (thus $I$ is no longer an initial strong component but the union of the initial strong components). If each initial strong component admits a nilpotent function over $A$, then $G[I]$ admits a nilpotent function $\tilde f$ over $A$. Let $\tilde\alpha\in A^I$ be such that $\tilde f^r=\cst=\tilde\alpha$ for some $r$. We will define a nilpotent $G$-function over $A$ by ``extending''~$\tilde f$.  

\medskip
Let $T$ be a spanning forest of $G$ ({\em i.e} $T$ is both a forest and a spanning subgraph of $G$) such that each root of $T$ belongs to $I$ (to get such a $T$ we can, for instance: consider the graph $G'$ obtained from $G$ by adding a new vertex $v$ and an arc from $v$ to each vertex in $I$; consider the spanning tree $T'$ of $G'$ obtained with a breadth-first-search starting from $v$; and set $T=T'\setminus v$). In this way, every vertex $i\notin I$ has thus a unique in-neighbor in $T$, denoted as $i^*$. For all $i\in[n]$, we denote by $\rho(i)$ the minimal length of a path of $T$ from $I$ to $i$ (thus $\rho(i)=0$ if and only if $i\in I$, and $\rho(i^*)<\rho(i)$ for all $i\not\in I$). Let $\alpha\in A^n$ be defined inductively as follows: for all $i$ such that $\rho(i)=0$, we set $\alpha_i=\tilde\alpha_i$; and for all $i$ with $\rho(i)>0$ we set
\begin{itemize}
\item
$\alpha_i=0$ if $\alpha_{i^*}=0$ and $i^*\in G^+(i)$, or $\alpha_{i^*}>0$ and $i^*\in G^-(i)$, or $\alpha_{i^*}\neq 1$ and $i^*\in G^0(i)$,
\item
$\alpha_i=1$ otherwise.  
\end{itemize}
Consider the function $f:A^n\to A^n$ defined by: 
\begin{itemize}
\item
for all $i\in I$, $f_i(x) =\tilde f_i(x_I)$,
\item
for all $i\notin I$ such that $\alpha_i=0$, 
\begin{equation}\label{eq:def3}
f_i(x) =\big(\bigwedge_{j\in G^+(i)}\One{x_j\geq 1}\big)\wedge\big(\bigwedge_{j\in G^-(i)}\One{x_j=0}\big)\wedge\big(\bigwedge_{j\in G^0(i)}\One{x_j=1}\big)
\end{equation}
\item
for all $i\notin I$ such that $\alpha_i=1$,
\[
f_i(x) =\big(\bigvee_{j\in G^+(i)}\One{x_j\geq 1}\big)\wedge\big(\bigvee_{j\in G^-(i)}\One{x_j=0}\big)\wedge\big(\bigvee_{j\in G^0(i)}\One{x_j\neq 1}\big).
\]  
\end{itemize}
Clearly, $f$ is a $G$-function if $s>1$ or if $s=1$ and all the unsigned arcs are in $G[I]$. Also, it is straightforward to prove, by induction on $\rho(i)$, that $f^{r+\rho(i)+k}_i(x)=\alpha_i$ for all $k\geq 0$, and consequently, $f^{r+p}=\cst=\alpha$, where $p=\max_{i\in[n]}\rho(i)$. This proves the proposition. 
\end{proof}

\begin{remark}\label{rem:initial_1}
If $|A|=2$ the additional condition on unsigned arcs is necessary, because some signed digraphs admit no boolean nilpotent functions while their initial strong components do. This is for instance the case with the following signed digraph $G$. Let $f$ be any boolean $G$-function. 
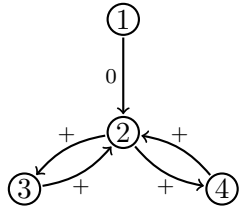
\[
\begin{tikzpicture}
\node[outer sep=1,inner sep=1,circle,draw,thick] (2) at (0,0){$2$};
\node[outer sep=1,inner sep=1,circle,draw,thick] (1) at (90:1.5){$1$};
\node[outer sep=1,inner sep=1,circle,draw,thick] (3) at (210:1.5){$3$};
\node[outer sep=1,inner sep=1,circle,draw,thick] (4) at (-30:1.5){$4$};
\path[thick]
(1) edge[->] node[inner sep=2,left]{\scriptsize $0$} (2)
(2) edge[->,bend right=20] node[inner sep=2,above]{\scriptsize $+$} (3)
(3) edge[->,bend right=20] node[inner sep=2,below]{\scriptsize $+$} (2)
(2) edge[->,bend right=20] node[inner sep=2,below]{\scriptsize $+$} (4)
(4) edge[->,bend right=20] node[inner sep=2,above]{\scriptsize $+$} (2)
;
\end{tikzpicture}
\]
Since vertex $1$ has no in-neighbor, we have $f_1=\cst=\alpha$ (thus the unique initial component trivially admits a boolean nilpotent function of class one). Also, we have necessarily $f_3(x)=f_4(x)=x_2$. Then, an analysis by cases shows that there are only two possibilities for $f_2$:  
\[
f_2(x)=(\overline{x_1}\land x_3)\lor (x_1\land x_4)
\quad\text{or}\quad
f_2(x)=(\overline{x_1}\land x_4)\lor (x_1\land x_3)
\]  
In the first case, for all $k\geq 2$ we have $f^k_2(x)=f^{k-1}_3(x)=f^{k-2}_2(x)$ if $\alpha=0$, and $f^k_2(x)=f^{k-1}_4(x)=f^{k-2}_2(x)$ otherwise. So $f_2^{k}=f_2^{k-2}$ for all $k\geq 2$, and we arrive to the same conclusion in the second case. Thus for every odd $k$ we have $f^k_2=f_2\neq\cst$. So $G$ admits no boolean nilpotent function. 
\end{remark}

\begin{remark}\label{rem:initial_2}
We deduce from the previous proposition that a digraph $G$ admits a boolean nilpotent function if and only if all its initial strong components do. 
\end{remark}

\section{Non-boolean nilpotent functions on signed digraphs}\label{sec:non-boolean_case}

Over an alphabet of size four, the question of the existence of nilpotent functions and their minimal class is easy.  

\begin{proposition}
Every signed digraph admits a nilpotent function over $\{0,1,2,3\}$ of class at most $2$. 
\end{proposition}
\begin{proof}
Let $G$ be a signed digraph on $[n]$ and let $f:\{0,1,2,3\}^n\to\{0,1,2,3\}^n$ be defined by:  
$$
f_i(x)=\big(\bigwedge_{j\in G^+(i)}\One{x_j\geq 2}\big)\wedge\big(\bigwedge_{j\in G^-(i)}\One{x_j<2}\big)\wedge\big(\bigwedge_{j\in G^0(i)}\One{x_j=2}\big)
$$
It is easy to check that $f$ is a $G$-function. Let $\alpha\in\{0,1\}^n$ be defined as follows: for all $i\in[n]$, $\alpha_i=0$ if $G^+(i)\cup G^0(i)\neq\emptyset$ and $\alpha_i=1$ otherwise. Then, for all $x\in \{0,1\}^n$, we have  $f(x)=\alpha$, and since $f(x)\in\{0,1\}^n$ for all $x\in\{0,1,2,3\}^n$ we deduce that $f^2(x)=\alpha$. 
\end{proof}

The first interesting case is thus the alphabet with three letters. Let us define a {\em balanced tree} as a signed tree $T$ in which there is at least one signed arc and at least one unsigned arc starting from each inner vertex of $T$. In a rooted tree, the {\em depth} of a vertex is the number of arcs in the path from the root to the vertex (thus the root has depth zero). The depth of a rooted tree is the maximal depth among its vertices. A rooted tree is {\em perfect} is all its leaves have the same depth (thus a perfect binary tree with of depth $d$ has $2^{d+1}-1$ vertices).

\medskip
Below, we prove that every signed digraph admit a nilpotent function over $\{0,1,2\}$ of logarithmic class. The proof is based on a decomposition of $G$ into balanced trees. The construction of the nilpotent $G$-function $f$ follows this decomposition, and the class of $f$ roughly corresponds to the maximum depth of a tree in this decomposition. In the following, logarithms are always in base two.

\begin{theorem}
Every signed digraph admits a nilpotent function over $\{0,1,2\}$ of class at most $\lfloor\log n\rfloor+2$. 
\end{theorem}

\begin{proof}
\setcounter{claimcounter}{0}
Let $G$ be a signed digraph on $[n]$. We first consider two special cases. Firstly, the case where $n=1$ is straightforward. Secondly, suppose that $G$ is acyclic and contains a perfect binary tree $T$ of depth $d$ as a spanning subgraph, with the possible addition of a loop on the root $r$. Let $f$ be any $G$-function such that $f^2_r(x)=\cst$. Since $f_r$ is either a constant or depends on $x_r$ only,  such a $G$-function exists by the first case. Let $\rho(i)$ the depth of each vertex $i$ in~$T$. By a straightforward induction on $\rho(i)$, we have $f_i^{\rho(i)+2}(x) = \cst$ for all $i\in[n]$. Thus $f$ is a nilpotent $G$-function over $\{0,1,2\}$ of class at most $\max_{i\in[n]}\rho(i)+2=d+2 \leq  \lfloor \log n \rfloor + 2$.

\medskip
We assume that $G$ does not fall in either of the special cases treated above. The proof is a construction involving the following objects.
\begin{enumerate}
\item 
Let $H$ be a maximal subgraph of $G$ that consists in a union of disjoint cycles, let $G'$ be the acyclic subgraph of $G$ obtained by removing all the arcs $(i,j)$ with $j\in V(H)$, and let $R$ be the set of sources of $G'$ (that is, the sources of $G$ plus the vertices of $H$). 
\item
Let $r_1,\dots,r_p$ be an enumeration of $R$, and let $T_1,\dots,T_p$ be a sequence of balanced trees constructed in the following way: 
\begin{itemize}
\item
$T_1$ is a maximal balanced tree with root $r_1$ contained in $G'$, 
\item
for $1< k \leq p$, $T_k$ is a maximal balanced tree with root $r_k$ contained in $G'\setminus (T_1\cup T_2\cup\dots\cup T_{k-1})$. 
\end{itemize}
If $G$ is not spanned by $T_1\cup\dots\cup T_p$, then let $T_{p+1},\dots,T_q$ be an additional sequence of be a balanced trees such that:
\begin{itemize}
\item
for $p<k \leq q$, $T_k$ is a maximal balanced tree contained in $G'\setminus (T_1\cup T_2\cup\dots\cup T_{k-1})$ such that $G'$ has an arc from a vertex $\ell_k$ in $T_1\cup T_2\cup\dots\cup T_{k-1}$ to the root $r_k$ of $T_k$,
\item
$G$ is spanned by $T=T_1\cup\dots\cup T_p\cup T_{p+1}\cup \dots\cup T_q$.
\end{itemize}
See Figure~\ref{fig} for an illustration. 
\begin{figure}
\[
\begin{tikzpicture}[scale=0.8]
\node[outer sep=1,inner sep=1,circle,draw,thick] (r1) at (4,6){$r_1$};
\node[outer sep=1,inner sep=1,circle,draw,thick] (r2) at (8,6){$r_2$};
\node[outer sep=1,inner sep=1,circle,draw,thick] (r3) at (12,6){$r_3$};
\node (1) at (2,4){$\bullet$};
\node (2) at (6,4){$\bullet$};
\node (3) at (10,4){$\bullet$};
\node (4) at (12,4){$\bullet$};
\node (5) at (14,4){$\bullet$};
%
\node (7) at (2,2){$\bullet$};
\node (8) at (4,2){$\bullet$};
\node[outer sep=1,inner sep=1,circle,draw,thick] (r6) at (8,2){$r_6$};
\node[outer sep=1,inner sep=1,circle,draw,thick] (r7) at (12,2){$r_7$};
\node[outer sep=1,inner sep=1,circle,draw,thick] (r5) at (2,0){$r_5$};
\node (9) at (6,0){$\bullet$};
\node (10) at (10,0){$\bullet$};
\node (11) at (12,0){$\bullet$};
\node[outer sep=1,inner sep=1,circle,draw,thick] (r8) at (14,0){$r_8$};

\path
(r1) edge[bend right=15,->] node[inner sep=2,below]{\scriptsize $+$} (r2)
(r2) edge[bend right=15,->] node[inner sep=2,above]{\scriptsize $-$} (r1)
(r1) edge[ultra thick,->] node[inner sep=4,left]{\scriptsize $+$} (1)
(r1) edge[ultra thick,->] node[inner sep=4,left]{\scriptsize $0$} (2)
(r2) edge[->] node[inner sep=4,left]{\scriptsize $0$} (2)
(r2) edge[->] node[inner sep=4,left]{\scriptsize $-$} (3)
(r3) edge[ultra thick,->] node[inner sep=4,left]{\scriptsize $0$} (3)
(r3) edge[ultra thick,->] node[inner sep=4,left]{\scriptsize $+$} (4)
(r3) edge[ultra thick,->] node[inner sep=4,left]{\scriptsize $-$} (5)
%
(1) edge[ultra thick,->] node[inner sep=4,left]{\scriptsize $+$} (7)
(1) edge[ultra thick,->] node[inner sep=4,left]{\scriptsize $0$} (8)
(2) edge[->] node[inner sep=4,left]{\scriptsize $-$} (8)
(2) edge[->] node[inner sep=4,left]{\scriptsize $+$} (9)
(2) edge[->] node[inner sep=4,left]{\scriptsize $+$} (r6)
(3) edge[->] node[inner sep=4,left]{\scriptsize $+$} (r6)
(3) edge[->] node[inner sep=4,left]{\scriptsize $+$} (r7)
(4) edge[->] node[inner sep=4,left]{\scriptsize $-$} (r7)
(5) edge[->] node[inner sep=4,left]{\scriptsize $-$} (r7)
(5) edge[->] node[inner sep=4,left]{\scriptsize $+$} (r8)
%
%
(7) edge[->] node[inner sep=4,left]{\scriptsize $+$} (r5)
(8) edge[->] node[inner sep=4,left]{\scriptsize $-$} (r5)
(8) edge[->] node[inner sep=4,left]{\scriptsize $+$} (9)
(r6) edge[ultra thick,->] node[inner sep=4,left]{\scriptsize $-$} (9)
(r6) edge[ultra thick,->] node[inner sep=4,left]{\scriptsize $0$} (10)
(r7) edge[->] node[inner sep=4,left]{\scriptsize $0$} (10)
(r7) edge[->] node[inner sep=4,left]{\scriptsize $+$} (11)
(r7) edge[->] node[inner sep=4,left]{\scriptsize $+$} (r8)
(9) edge[->] node[inner sep=4,below]{\scriptsize $+$} (r5)
(10) edge[bend left,->] node[inner sep=4,below]{\scriptsize $+$} (r5)
%
(11) edge[->] node[inner sep=4,above]{\scriptsize $-$} (r8)
(r8) edge[bend left,->] node[inner sep=4,below]{\scriptsize $-$} (10)
(r5) edge[bend left=100,->] node[inner sep=4,left]{\scriptsize $-$} (r1)
;
\end{tikzpicture}
\qquad\quad
\]
{\caption{\label{fig} An example of decomposition $T$. The thick arcs are those that belong to $T$. $H$ consists in the cycle of length two between $r_1$ and $r_2$, and $r_3$ is the unique source, so that $R=\{r_1,r_2,r_3\}$. Thus $p=3$ and $q=8$. The balanced trees $T_2$, $T_5$, $T_7$ and $T_8$ are reduced to a single vertex.}}
\end{figure}
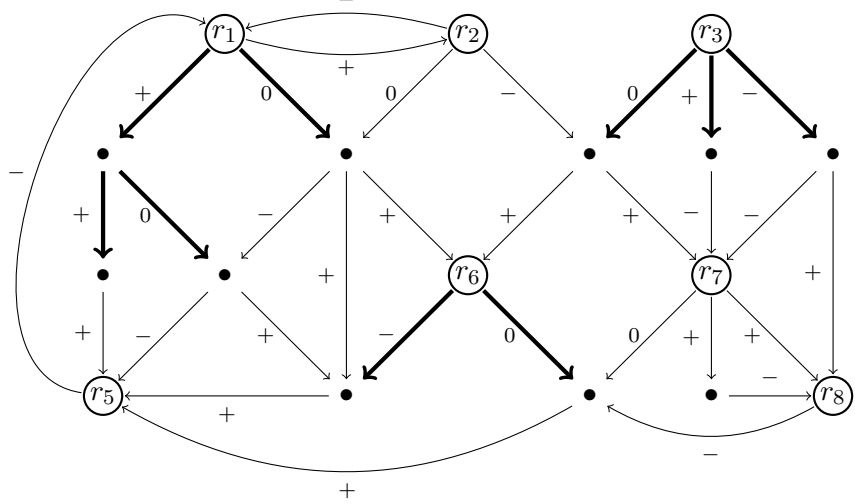
\item
Every vertex $i$ that is not a source of $T$ has a unique in-neighbor in $T$ that we denote as~$i^-$. Also, every vertex $i$ in $H$ has a unique in-neighbor, that we denote as $i^-$, and a unique out-neighbor, that we denote as $i^+$. If $i=r_k$ for some $p< k\leq q$, then we set $i^-=\ell_k$. In this way, $i^-$ is well defined for every vertex $i$ of $G$ that is not a source of $G$. 
\item
For every vertex $i$, let $P_i$ be a path of $T$ of minimal length from $i$ to a leaf of $T$, and let $\sigma_i$ be the out-neighbor of $i$ in $P_i$.
\item
For all $k\geq 0$, we denote by $M_k$ the set of vertices $i$ with depth $k$ in $T$. Thus, $M_0$ contains exactly the roots $r_1,\dots,r_q$, and if $i\in M_k$ with $k>0$ then $i^-\in M_{k-1}$. We define the labeling function $\rho:V(T)\to \N$ as follows:
\begin{itemize}
\item
for all $i\in M_0$, $\rho(i)=0$, 
\item
for all $k>0$ and $i\in M_k$, if $(i^-,i)$ and $(i^-,\sigma_{i^-})$ are both unsigned or both signed then $\rho(i)=\rho(i^-)+1$, and otherwise $\rho(i)=0$. 
\end{itemize}
\item
Let $\alpha\in \{1,2\}^n$ be defined by:
\begin{itemize}
\item
if $i$ is a source of $G$ then $\alpha_i=1$, 
\item
if $i$ is a vertex of $H$ then $\alpha_i=2$ if and only if $(i,i^+)$ is unsigned,
\item
if $i$ is a leaf of $T\setminus R$ then $\alpha_i=2$ if and only if $i=\ell_k\in G^0(r_k)$ for some $p< k\leq q$,
\item 
if $i$ is not a leaf of $T\setminus R$ then $\alpha_i=2$ if and only if $T$ has a signed arc $(i,j)$ with~$\rho(i)<\rho(j)$.
\end{itemize}
\item
Let $S$ be the set of vertices $i$ with the following properties: $\rho(i)=0$, $i$ is not a source of $G$, and either $(i^-,i)$ is signed and $\alpha_{i^-}=2$ or $(i^-,i)$ is unsigned and $\alpha_{i^-}=1$. Note that $S\cap R=\emptyset$.
\item
Note that if $\rho(i)>0$ then $\rho(i)=\rho(i^-)+1$ and this property allows to define $\beta\in\{0,1\}^n$ inductively as follows: 
\begin{itemize}
\item
if $\rho(i)=0$ and $i\not\in S$ then $\beta_i=0$ if and only if $i$ is a source of $G$ or $(i^-,i)$ is positive,
\item 
if $\rho(i)>0$ or $\rho(i)=0$ and $i\in S$, then $\beta_i=0$ if and only if one of the following conditions is satisfied: $(i^-,i)$ is positive and $\alpha_{i^-}\beta_{i^-} <2$; $(i^-,i)$ is negative and $\alpha_{i^-}\beta_{i^-} =2$; $(i^-,i)$ is unsigned and $\alpha_{i^-}\beta_{i^-}\neq 1$.
\end{itemize}
\item
Finally, let $f:\{0,1,2\}^n\to\{0,1,2\}^n$ be defined as follows: for all $i\in [n]$, 
\begin{itemize}
\item
if $i$ is a source of $G$ then $f_i(x)=0$,
\item 
if $i$ is not a source of $G$ and $\beta_i=0$ then 
$$
f_i(x)=\alpha_i\left[\big(\bigwedge_{j\in G^+(i)}\One{x_j=2}\big)\wedge\big(\bigwedge_{j\in G^-(i)}\One{x_j<2}\big)\wedge\big(\bigwedge_{j\in G^0(i)}\One{x_j=1}\big)\right],
$$
\item
if $i$ is not a source of $G$ and $\beta_i=1$ then
$$
f_i(x)=\alpha_i\left[\big(\bigvee_{j\in G^+(i)}\One{x_j= 2}\big)\vee\big(\bigvee_{j\in G^-(i)}\One{x_j<2}\big)\vee\big(\bigvee_{j\in G^0(i)}\One{x_j\neq1}\big)\right].
$$
\end{itemize}
\end{enumerate}

It is easy to see that $f$ is a $G$-function, and we will prove that $f$ is a nilpotent function of class at most $\lfloor \log n\rfloor+2$.

\begin{claim}
$f$ is a nilpotent function of class at most $\max_{i\in [n]}\rho(i) + 3$. 
\end{claim}
\begin{subproof}
Actually, we will prove that for all $i\in [n]$, $x\in\{0,1,2\}^n$ and $k\in\N$,  
\begin{equation*}
\label{eq:main}
f^{\rho(i)+3+k}_i(x)=\alpha_i\beta_i.
\end{equation*}
We proceed by induction on $\rho(i)$. 

\medskip
{\em Case $\rho(i)=0$ and $i \notin S$.} We prove $f^{2+k}_i(x)=\alpha_i\beta_i$ which is stronger. If $i$ is a source of $G$ then $\beta_i=0$ and $f_i(x)=0=\alpha_i\beta_i$ for all $x$. If $i$ is not a source of $G$ then we have three cases:
\begin{enumerate}
\item
If $(i^-,i)$ is unsigned then $\beta_i=1$ and since $i\not\in S$, $\alpha_{i^-}=2$. Thus $f_{i^-}(x)\in \{0,2\}$ and $x_{i^-}\neq 1\Rightarrow f_i(x)=\alpha_i$ thus $f^{2+k}_i(x)=\alpha_i\beta_i$. 
\item
If $(i^-,i)$ is positive then $\beta_i=0$ and since $i\not\in S$, $\alpha_{i^-}=1$. Thus $f_{i^-}(x)\in \{0,1\}$ and $x_{i^-}<2\Rightarrow f_i(x)=0$, and we deduce that $f^{2+k}_i(x)=0=\alpha_i\beta_i$.
\item
If $(i^-,i)$ is negative then $\beta_i=1$ and since $i\not\in S$, $\alpha_{i^-}=1$. Thus $f_{i^-}(x)\in \{0,1\}$ and $x_{i^-}<2\Rightarrow f_i(x)=\alpha_i$ thus $f^{2+k}_i(x)=\alpha_i\beta_i$.
\end{enumerate}
In any case, $f_i^{2+k}(x) = \alpha_i \beta_i$. 

\medskip
{\em Case $i \in S$.} Note that $\rho(i) = 0$. Let us prove that $i^-\in R$. Suppose, for a contradiction, that $i^-\not\in R$. Then we consider two cases:
\begin{enumerate}
\item
Suppose that $i^-$ is a leaf of $T\setminus R$. Then $(i^-,i)=(\ell_k,r_k)$ for some $p<k\leq q$. If $(i^-,i)$ is unsigned then, by the definition of $\alpha$, we have $\alpha_{i^-}=2$, and this not possible since $i\in S$. If $(i^-,i)$ is signed we have $\alpha_{i^-}=2$ (by the fact that $i\in S$), and according to the definition of $\alpha$, $i^-=\ell_{k'}$ for some $p<k'\leq q$ such that $(\ell_{k'},r_{k'})$ is unsigned. We deduce that $i^-=\ell_k=\ell_{k'}$ is a leaf of a tree $T_m$ with $m<\min(k,k')$. But then, by adding to $T_m$ the signed arc $(i^-,r_k)$ and the unsigned arc $(i^-,r_{k'})$ we obtain a balanced tree $T'_m$ contained in $G'\setminus (T_1\cup T_2\cup\dots\cup T_{m-1})$ which is greater than $T_m$, another contradiction. 
\item
Suppose that $i^-$ is not leaf of a $T\setminus R$. Then $(i^-,i)$ is an arc of $T$. If $(i^-,i)$ is unsigned, then $T$ has a signed arc $(i^-,j)$, and since $\rho(i)=0$, we deduce that $\rho(i^-)<\rho(j)$. Thus $\alpha_{i^-}=2$, and this is not possible since $i\in S$. Similarly, if $(i^-,i)$ is signed then for every signed arc $(i^-,j)$ of $T$ we have $\rho(j)=\rho(i)=0$, and we deduce that $\alpha_i=1$, and this is not possible since $i\in S$. 
\end{enumerate}
Thus $i^-\in R$ hence $\rho(i^-) = 0$ and $i^- \notin S$, thus following the previous case,
\begin{equation}\label{eq:induction}
f_{i^-}^{2+k}(x) = \alpha_{i^-} \beta_{i^-}.
\end{equation}
Suppose that $(i^-,i)$ is positive. 
\begin{enumerate}
\item
If $\alpha_{i^-}\beta_{i^-} <2$ then $\beta_i=0$ thus $x_{i^-}<2\Rightarrow f_i(x)=0$. We deduce from (\ref{eq:induction}) that 
$$
f^{3 + k}_i(x)=0 = \alpha_i \beta_i.
$$
\item
If $\alpha_{i^-}\beta_{i^-}=2$ then $\beta_i=1$ thus $x_{i^-}=2\Rightarrow f_i(x)=\alpha_i$. We deduce from (\ref{eq:induction}) that 
$$
f^{3+k}_i(x)=\alpha_i = \alpha_i \beta_i.
$$
\end{enumerate}
If $(i^-,i)$ is negative or unsigned we prove with similar arguments that $f^{3+k}_i(x)=\alpha_i\beta_i$.

\medskip
{\em Case $\rho(i)>0$.} We then have $\rho(i^-)=\rho(i)-1$. Following the induction hypothesis we have  $f_{i^-}^{\rho(i) + 2 +k}(x) = f^{\rho(i^-)+ 2 +k}_{i^-}(x)=\alpha_{i^-}\beta_{i^-}$ and to complete the induction step, we proceed as above, proceeding by case on the sign of $(i^-,i)$. 
\end{subproof}

It remains to prove the following claim. 

\begin{claim}
$\max_{i\in[n]} \rho(i) + 1\leq\lfloor \log n\rfloor$. 
\end{claim}
\begin{subproof}
We need the following:
\begin{equation}
\label{eq:perfect}
\text{Every vertex $i$ is the leaf of a perfect balanced binary tree of depth $\rho(i)$ contained in $T$.}
\end{equation}
We prove this by induction on $\rho(i)$. This is obvious if $\rho(i)=0$. Suppose that $\rho(i)>0$ and consider the path $i_0,i_1,\dots,i_d$ of $T$ such that $i_d=i$ and $\rho(i_k)=k$ for all $0\leq k\leq d$. By induction hypothesis, $i_{d-1}$ is the leaf of a perfect balanced binary tree $B$ of depth $d-1$ contained in $T$ (and $i_0$ is the root of $B$). Let $L$ be the set of leaves of $B$ and let $L'$ be the set of leaves of $B$ that are also leaves of $T$ (note that $i_{d-1}\in L\setminus L'$). 
\begin{enumerate}
\item
Suppose that $L'=\emptyset$. Then for each $j\in L$, $T$ contains two arcs starting from $j$, say $a_j$ and $b_j$, such that $a_j$ and $b_j$ are not both signed or both unsigned; and we can choose $a_{i_{d-1}}$ to be the arc from $i_{d-1}$ to $i_d$. Let $B'$ be the tree obtained by adding the $2|L|$ arcs $a_j,b_j$. Then $B'$ is a perfect balanced binary tree of depth $d$, and $i_d$ is a leaf  of $B'$. Thus (\ref{eq:perfect}) holds. 
\item
Suppose that $L'\neq\emptyset$. Then there exists a greatest index $0\leq t<d-1$ such that $B$ contains a path from $i_t$ to $L'$. Then $P_{i_t}$ is necessarily a path from $i_t$ to $L'$. Thus $i_{t+1}$ is not in this path, so $(i_t,\sigma_{i_t})$ and $(i_t,i_{t+1})$ cannot be both signed or both unsigned and we deduce that $\rho(i_{t+1})=0$, a contradiction. This proves (\ref{eq:perfect}). 
\end{enumerate}
We are now in position to prove the claim. Let $i$ be such that $\rho(i)$ is maximal. According to~(\ref{eq:perfect}), $i$ is the leaf of a perfect balanced binary tree $B\subseteq T$ of depth $\rho(i)$. Thus $B$ has $m=2^{\rho(i)+1}-1$ vertices, so $\rho(i)+1 \leq\lfloor\log n\rfloor$ unless $n = m$. In that case, $B$ is a spanning subgraph of $G$, hence $|R| = 1$, thus $G$ falls in the second special case treated at the beginning of the proof. 
\end{subproof}  
\end{proof}  

The following proposition shows that the bound $\lfloor \log n\rfloor+2$ is tight. 

\begin{proposition}
For every $n\geq 1$, there exists a signed digraph $G$ on $[n]$ such that every nilpotent $G$-function over $\{0,1,2\}$ is of class at least $\lfloor \log n\rfloor+2$. Furthermore, if $n=2^{\lfloor \log n\rfloor+1}-1$, then there exists a strong signed digraph with this property.
\end{proposition}
\begin{proof}
Let $d=\lfloor \log n\rfloor$, and suppose first that $n=2^{d+1}-1$. Let $T$ be a perfect balanced binary tree of depth $d$ with vertex set $[n]$. Let $T'$ be the strong signed digraph obtained from $T$ by adding an arc (of any sign) from every leaf to the root $r$ of $T$. Let $f$ be a nilpotent $T'$-function over $\{0,1,2\}$ and for each $i\in [n]$, let $k_i$ be the smallest integer such that $f^{k_i}_i=\cst$. Let $i$ be a non-leaf of $T$, and let $(i,j)$ and $(i,\ell)$ be the signed and unsigned arcs starting from $i$. Since $f_j(x)$ and $f_\ell(x)$ only depend on $x_i$, we abusively write $f_j(x_i)$ and $f_\ell(x_i)$. Obviously, we have $\max(k_j,k_\ell)\leq k_i+1$, but actually 
\begin{equation}\label{eq:max}
\max(k_j,k_\ell)=k_i+1.
\end{equation}
Indeed, since $(i,j)$ is signed, $f_j$ is monotonous thus $f_j(0)\neq f_j(2)$. It means that if $\{0,2\}$ is a subset of the image of $f^{k_i-1}_i$, then $f^{k_i}_j$ is not a constant thus $k_j=k_i+1$. Then, since $(i,\ell)$ is unsigned, $f_\ell$ is non-monotonous thus $f_\ell(0)\neq f_\ell(1)$ and $f_\ell(1)\neq f_\ell(2)$. It means that if $\{0,1\}$ or $\{1,2\}$ is a subset of the image of $f^{k_i-1}_i$ then $k_\ell=k_i+1$. Since $f^{k_i-1}_i$ is not a constant, either $\{0,1\}$ or $\{1,2\}$ or $\{0,2\}$ is a subset of $\text{Im}(f^{k_i-1}_i)$ thus $k_j=k_i+1$ or $k_\ell=k_i+1$ and (\ref{eq:max}) follows. Then, we deduce that there exists a leaf $\ell$ of $T$ with $k_{\ell}=k_r+d$. Since $f_r$ is not a constant, we have $k_r\geq 2$ thus $f$ is of class at least $d+2$. 

\medskip
Suppose now that $n<2^{d+1}-1$. Let $T$ be a perfect balanced binary tree of depth $d-1$ with root~$r$. Let $T'$ be the signed digraph obtained from $T$ by 
\begin{enumerate}
\item
adding a new vertex $w$ and $p=n-2^d$ other vertices $v_1,\dots,v_p$ (so that $T'$ has $n$ vertices);  
\item
adding a positive loop on $w$ and an unsigned arc from $w$ to every vertex of $T$ and every vertex $v_q$, $1\leq q\leq p$.
\end{enumerate}
Let $f$ be a nilpotent $T'$-function, and for each vertex $i$, let $k_i$ be as previously. Using the arguments above, we show that there exists a leaf $\ell$ of $T$ such that $k_\ell=k_r+d-1$. It is thus sufficient to prove that $k_r\geq 3$. Since $f_w$ only depends on $x_w$ and is non-decreasing with $x_w$, if $\{0,2\}\subseteq \text{Im}(f_w)$ then we have  $f_w(0)=0$ and $f_w(2)=2$ thus $f$ is not nilpotent. Thus $\{0,1\}$ or $\{1,2\}$ is a superset of $\text{Im}(f_w)$. Since $f_r$ is a non-monotonous function of $x_w$, we have $f_r(0)\neq f_r(1)$ and $f_r(1)\neq f_r(2)$ thus $f^2_r$ is not a constant and we deduce that $k_r\geq 3$. 
\end{proof}

In the above constructions showing that the bound $\lfloor \log n\rfloor+2$ is tight, some vertices have a unique unsigned predecessor, while some others have no unsigned predecessor. The following proposition shows that if at least one of this two conditions fails, then there exists a nilpotent function of class 2, as in the four letters case. In particular, if all arcs are labelled positively or negatively, then there exists a nilpotent function of class 2.  
    
\begin{proposition}
Let $G$ be a signed digraph. If $G$ has no vertices with a unique unsigned predecessor, or if all the vertices of $G$ have at least one unsigned predecessor, then $G$ admits a nilpotent function over $\{0,1,2\}$ of class at most $2$.
\end{proposition}
\begin{proof}
Let $G$ be a signed digraph on $[n]$. Suppose first that $G$ has no vertices with a unique unsigned  predecessor. Let $f$ be the $G$-function over $\{0,1,2\}$ defined for all $i\in [n]$ by: $$
f_i(x)=\big(\bigvee_{j\in G^+(i)}\One{x_j= 2}\big)\vee\big(\bigvee_{j\in G^-(i)}\One{x_j<2}\big)\vee\big(\bigoplus_{j\in G^0(i)}\One{x_j=2}\big)
$$
Note that $f$ is a $G$-function because there is no $i$ such that $|G^0(i)|=1$. We have $f(x)\in\{0,1\}^n$ for all $x\in\{0,1,2\}^n$, and for all $x\in\{0,1\}^n$, we have $f_i(x)=0$ if $G^-(i)=\emptyset$ and $f_i(x)=1$ otherwise. Thus $f$ is a nilpotent function of class 2.

\medskip
Suppose now that each vertex has at least one unsigned predecessor. Let $f$ be the $G$-function over $\{0,1,2\}$ defined for all $i\in [n]$ by:
$$
f_i(x)=2\left[\big(\bigwedge_{j\in G^+(i)}\One{x_j =2}\big)\wedge\big(\bigwedge_{j\in G^-(i)}\One{x_j<2}\big)\wedge\big(\bigwedge_{j\in G^0(i)}\One{x_j=1}\big)\right]
$$
We have $f(x)\in\{0,2\}^n$ for all $x\in\{0,1,2\}^n$, and for all $x\in\{0,2\}^n$ we have $f_i(x)=0$  since $|G^0(i)|\geq 1$. Thus $f$ is a nilpotent function of class at most 2.
\end{proof}

\section{Boolean nilpotent functions on unsigned digraph}\label{sec:boolean_case}

In contrast with the non-boolean case, the following question is difficult. 

\begin{question}
Which signed digraphs admit a boolean nilpotent function? And if a signed digraph admits a boolean nilpotent function, what is its minimal class?  
\end{question}

Foremost, as shown by Propositions~\ref{pro:existence} and \ref{pro:strong}, some signed digraphs admit no boolean functions, and some signed digraphs admit no boolean nilpotent functions while all their initial strong components do (cf. Remark~\ref{rem:initial_1}). In addition, some strong signed digraphs admit no boolean nilpotent functions. For instance, if $G$ is a strong signed digraph with only negative (resp. positive) cycles, then every boolean $G$-function has no fixed points (resp. at least two fixed points) and is thus not nilpotent \cite{A08}. These observations lead us to study the unsigned version of the question, which is more tractable. 

\medskip
Indeed, in the unsigned case, every digraph admits a boolean function and a digraph admits a boolean nilpotent function if and only if all its initial strong components do (cf. Remark~\ref{rem:initial_2}). These are helpful simplifications. However, there are still, in the unsigned case, some digraphs that admit no boolean nilpotent functions. The most simple example is the directed cycle (if the interaction graph of a boolean function $f$ is a directed cycle, then $f$ has no fixed points if the cycle is negative, and two otherwise). In the following, we exhibit families of digraphs that admit a boolean nilpotent function whose class is at most a quadratic or linear function of $n$ (Section~\ref{sec:linear}). We then exhibit families of digraphs that admit a boolean nilpotent function of constant class (Section~\ref{sec:constant}).

\subsection{Nilpotent functions of non-constant class}\label{sec:linear}

A {\em walk} in a digraph $G$ is a sequence $i_0,\dots,i_p$ of vertices of $G$ such that $(i_k,i_{k+1})$ is an arc of $G$ for all $0\leq k<p$. Such a walk is a walk of length $p$ from $i_0$ to $i_p$. If $G$ is strong, the {\em loop number} of a $G$ is the greatest common divisor of the length of the cycles of~$G$ \cite{CLP04}. We say that $G$ is {\em primitive} if $G$ is strong and has loop number one. The following theorem is the graph-theoretic analogue of a classical result on powers of non-negative matrices.

\begin{theorem}[\cite{HV58}]\label{thm:primitive}
Let $G$ be a primitive digraph with $n$ vertices, and let $u$ and $v$ be two vertices (not necessarily distinct). For every $p\geq n^2-2n+2$, there exists a walk from $u$ to $v$ of length $p$.
\end{theorem}

Suppose that $G$ is a strong digraph on $[n]$, and let $f$ be the disjunctive network on $G$, that is, the $G$-function defined by $f_i(x)=\vee_{j\in G(i)} x_j$ for all $i\in [n]$. It is well known that the length of every limit cycle of $f$ divides the loop number of $G$ \cite{CLP04,SLV10,GM12}. Thus, in particular, if $G$ is primitive, then all the limit cycles of $f$ are fixed points, and it results from the theorem above that $f$ reaches a fixed point in at most $n^2-2n+2$ steps. This bound is optimal \cite{HV58,SM99}. However, $f$ is not nilpotent, since $0$ and $1$ are fixed points of $f$. By considering the conjunctive network associated with $G$, and defined in a similarly way, all these  observations remain valid. Below, we show that by adding at least one arc in $G$, one can break one of the two fixed points, and obtained a boolean nilpotent function of class at most $n^2-2n+3$. We have not been able to prove that this bound is optimal.

\begin{proposition}\label{pro:primitive}
Let $G$ be a digraph with $n$ vertices. If $G$ has a primitive spanning strict subgraph, then $G$ admits a boolean nilpotent function of class at most $n^2-2n+3$.
\end{proposition}

\begin{proof}
Suppose that $G$ has a primitive spanning strict subgraph $H$. Let $f$ be the boolean $G$-function defined by for all $i\in [n]$ by:
\[
f_i(x)=\bigwedge_{j\in H(i)}x_j\land \bigwedge_{j\in G(i)\setminus H(i)}\overline{x_j}.
\]
Let $x\in\{0,1\}^n$ and $p\geq n^2-2n+2$. Suppose that $x_i=0$ for some $i$, and let $j$ be any vertex of~$G$. By Theorem~\ref{thm:primitive}, $H$ has a walk from $i$ to $j$ of length $p$, and since $x_i=0$, we deduce from the definition of $f$ that $f^p_j(x)=0$. Thus if $x<1$ then  $f^p(x)=0$. Now, if $x=1$ then $f_i(x)=0$ for every $i$ such that $G(i)\setminus H(i)\neq\emptyset$, and such $i$ exists since $H$ is a strict subgraph of $G$. Thus $f(x)<1$ and we deduce that $f^{p+1}(x)=0$. This proves that $f^{n^2-2n+3}=\cst=0$. 
\end{proof}

The theorem below shows, in a strong sense, that if $G$ is itself primitive but has no primitive spanning strict subgraphs, then $G$ does not necessarily admit a boolean nilpotent function. Let us call {\em double-cycle} the digraph $C_{\ell,r}$ obtained from two cycles $C_\ell$ and $C_r$, of length $\ell$ and $r$ respectively, by identifying one vertex. The next theorem characterizes the double cycles that admit a boolean nilpotent function as well as the class of such nilpotent functions (see \cite{DNS2012} for an analysis of limit cycles of boolean $C_{\ell,r}$-functions).

\begin{theorem}
$C_{\ell,r}$ admits a boolean nilpotent function if and only if $\min(\ell,r)$ divides $\max(\ell,r)$, and the class of every boolean nilpotent $C_{\ell,r}$-function is $2\max(\ell,r)-1$. 
\end{theorem}

\begin{proof}
Let $1,2,\dots,\ell$ be the vertices of $C_\ell$ given in the order, and let $1,\ell+1,\ell+2,\dots,\ell+r-1$ be the vertices of $C_r$ given in the order. Let $f$ by any boolean $C_{\ell,r}$-function and $n=\ell+r-1$. For every $b\in\{0,1\}^n$, the function $h$ that maps every $x\in\{0,1\}^n$ to $f(x\oplus b)\oplus b$ is a boolean $C_{\ell,r}$-function isomorphic to $f$ (thus $h$ is nilpotent of class $k$ if and only if $f$ is). Hence, we can assume, without loss of generality, that 
\[
\left\{
\begin{array}{l}
f_2(x)=f_{\ell+1}(x)=x_1\\
f_i(x)=x_{i-1}\text{ for all }i\not\in\{1,2,\ell+1\}.
\end{array}
\right.
\] 
In this way, for all $k\in\N$ we have 
\begin{equation}\label{eq:double_cycle}
\left\{
\begin{array}{l}
f^{k+i}_{1+i}(x)=f^k_1(x)\text{ for }1\leq i< \ell\\[1mm]
f^{k+i}_{\ell+i}(x)=f^k_1(x)\text{ for }1\leq i<r.
\end{array}
\right.
\end{equation}
The only component of $f$ that is not defined is $f_1$, which only depends on $x_\ell$ and $x_n$. We proceed by cases. If $f_1(x)=x_\ell\land x_n$ or $f_1(x)=x_\ell\lor x_n$ then the interaction graph of $f$ has only positive cycles, thus $f$ has at least two fixed points, and thus $f$ is not nilpotent.   Also, if $f_1(x)=\overline{x_\ell}\land\overline{x_n}$ or $f_1(x)=\overline{x_\ell}\lor\overline{x_n}$, then the interaction graph of $f$ has only negative cycles, thus $f$ has no fixed points, and thus $f$ is not nilpotent. It remains the six cases below. In the first two ones, $f$ is nilpotent if and only if $\ell=r$. In the other four cases, $f$ is nilpotent if and only if $\ell$ and $r$ are not coprime. In every case, when $f$ is nilpotent its class is $2\max(\ell,r)-1$. 
\begin{description}
\item[{\it Case 1: $f_1(x)=x_\ell\oplus x_n$.}] 
We prove the following: $f$ is nilpotent if and only if $\ell=r$, and if $f$ is nilpotent, then its class is $2\ell-1=n$. 

Suppose that $f$ is nilpotent with class $m$ and suppose, without loss of generality, that $r\geq\ell$. Then $m\geq r$ and since $f(0)=0$ we have $f^m=0$. For all $0\leq k<\ell$, let $Y^k$ be the set of $y\in\{0,1\}^n$ such that 
\[
\left\{
\begin{array}{l}
y_{\ell-k}=1\text{ and }y_i=0\text{ for }1\leq i<\ell-k\\
y_{n-k}=1\text{ and }y_i=0\text{ for }\ell+1\leq i<n-k
\end{array}
\right.
\]
It is easy to see that: if $f(x)=0$ and $x\neq 0$ then $x\in Y^0$; and if $f(x)\in Y^k$ then $x\in Y^{k+1}$. We deduce that for each $x$ such that $f^{m-1}(x)\neq 0$ we have $f^{m-\ell+1}(x)\in Y^{\ell-2}$. Thus we have $f^{m-\ell+1}_2(x)=1$. Now, if $\ell<r$ then $\ell+1<n-\ell+2$ thus by the definition of $Y^{\ell-2}$ we have $f^{m-\ell+1}_{\ell+1}(x)=0$. But according to (\ref{eq:double_cycle}) we have $f^{m-\ell+1}_2(x)=f^{m-\ell}_1(x)=f^{m-\ell+1}_{\ell+1}(x)$, a contradiction. We deduce that $\ell=r$.  

Suppose now that $\ell=r$, and let us prove that $f$ is a nilpotent function of class $2\ell-1$. Following (\ref{eq:double_cycle}) we have 
\[
\forall k\geq\ell,\qquad f^k_1(x)=f^{k-1}_\ell(x)+f^{k-1}_{2\ell-1}(x)=f^{k-\ell}_1(x)+f^{k-\ell}_1(x)=0
\]
and we deduce from this and (\ref{eq:double_cycle}) that 
\[
\forall 1\leq i<\ell,\qquad f^{2\ell-1}_{1+i}(x)=f^{2\ell-1}_{\ell+i}(x)=f^{2\ell-1-i}_1(x)=0.
\]
Thus $f^{2\ell-1}(x)=0$ so $f$ is a nilpotent function of class at most $2\ell-1$. Let $z\in\{0,1\}^n$ be defined by $z_i=1$ if and only if $1\leq i\leq\ell$. Then $f^{\ell-1}(z)=1$ thus $f^{2\ell-2}_\ell(z)=f^{2\ell-2}_{2\ell-1}(z)=1$ and we deduce that the class of $f$ is exactly $2\ell-1$. 
\item[{\it Case 2: $f_1(x)=x_\ell \oplus x_n \oplus 1$.}]
We prove with similar arguments that $f$ is nilpotent if and only if $\ell=r$, and that if $f$ is nilpotent then its class is $2\ell-1$. 
\item[{\it Case 3: $f_1(x)=x_\ell \land \overline{x_n}$.}]
We prove the following: $f$ is nilpotent if and only if $\ell$ divides $r$; and if $f$ is nilpotent, then its class is $2r-1$. Foremost, we have following properties
\begin{equation}\label{eq:and_case}
\begin{array}{ll}
\forall k,p\geq 0,&\qquad f^k_1(x)=0~\Rightarrow f^{k+p\ell}_1(x)=0,\\[1mm]
\forall k\geq r,&\qquad f^k_1(x)=1~\Rightarrow f^{k-r}_1(x)=0.\\[1mm]
\end{array}
\end{equation}
Suppose that $\ell$ divides $r$, and let $p$ be such that $r=p\ell$. We deduce from the implications above that if $f^k_1(x)=1$ with $k\geq r$ then $f^{k-r}_1(x)=0$ and thus $f^{k-r+p\ell}_1(x)=0$, a contradiction. Thus  $f^k_1(x)=0$ for all $k\geq r$, and we deduce that $f^{2r-1}(x)=0$. Thus $f$ is a nilpotent function of class at most $2r-1$. Let $z\in\{0,1\}^n$ be defined by $z_i=1$ if and only if $1\leq i\leq \ell$. We have $f^{r-1}(z)=1$, and thus $f^{2r-1}_{\ell+r-1}(z)=1$. Hence, the class of $f$ is exactly $2r-1$. 

Suppose now that $\ell$ does not divide $r$. Let $z$ be such that $z_1=1$ and $z_i=0$ for $i\neq 1$. Following~(\ref{eq:and_case}) we have 
$$
\forall p\in\N,~\forall 1\leq k<\ell,\qquad f^{p\ell+k}_1(z)=0.
$$
We now prove that $f^{p\ell}_1(z)=1$ by induction on $p$. We have $f^0_1(z)=z_1=1$. Let $p>1$ and suppose that $f^{p\ell-\ell}_1(z)=1$. Then $f^{p\ell-1}_\ell(z)=1$. If $p\ell>r$ then $f^{p\ell-1}_n(z)=0$ since otherwise we have $f^{p\ell-r}_1(z)=1$ thus $p\ell-r$ is a multiple of $\ell$ and thus $r$ is a multiple of $\ell$, a contradiction. Hence, $f^{p\ell-1}_\ell(z)=1$ and $f^{p\ell-1}_n(z)=0$ thus $f^{p\ell}_1(z)=1$. So we have proved that $f^k_1(z)=1$ if and only if $k$ is a multiple of $\ell$, thus $f$ is not nilpotent.        
\item[{\it Case 4: $f_1(x)=\overline{x_\ell} \land x_n$.}]
We prove with similar arguments that $f$ is nilpotent if and only if $r$ divides $\ell$, and that if $f$ is nilpotent, then its class is $2\ell-1$.
\item[{\it Case 5: $f_1(x)=x_\ell \lor \overline{x_n}$.}] 
We prove with similar arguments that $f$ is nilpotent if and only if $\ell$ divides $r$, and that if $f$ is nilpotent, then its class is $2r-1$.
\item[{\it Case 6: $f_1(x)=\overline{x_\ell} \lor x_n$.}]
We prove with similar arguments that $f$ is nilpotent if and only if $r$ divides $\ell$, and that if $f$ is nilpotent, then its class is $2\ell-1$.
\end{description} 
\end{proof}

For all $m\geq 1$, the {\em wheel} $W_m$, also called {\em $m$-wheel}, is the digraph obtained from $C_m$ by adding a vertex $v$, called the {\em center}, and an arc from $v$ to every vertex of $C_m$. We say that $G$ contains an $m$-wheel if some subgraph of $G$ is isomorphic to $W_m$. So for instance, if $G$ is strong and has at least two vertices, then $G$ contains a loop if and only if it contains $W_1$. Note also that $G$ contains a wheel if and only if the out-neighborhood of some vertex induces a non-acyclic digraph.

\medskip
Let $G$ be a strong digraph with $n\geq 2$ vertices. Below, we prove that if $G$ has an $m$-wheel, then $G$ admits a boolean nilpotent function of class at most $2n-m+1$. As a consequence, if $G$ has a loop then $G$ admits a nilpotent function of class at most $2n$. For this particular case, we establish a better bound, $2n-1$, which is optimal, as shown by the double cycle $C_{1,n}$ and the previous theorem. 

\begin{theorem}\label{thm:loop} Let $G$ be a strong digraph with $n\geq 2$ vertices. 
\begin{itemize}
\item
If $G$ has a $m$-wheel, then $G$ admits a boolean nilpotent function of class at most $2n-m+1$. 
\item
If $G$ has a loop, then $G$ admits a boolean nilpotent function of class at most $2n-1$.
\end{itemize}
\end{theorem}

\begin{remark}
If $G$ is strong and has a $m$-wheel with $m\geq 2$, then the fact that $G$ admits a boolean nilpotent function is an easy consequence of Proposition~\ref{pro:primitive} (but the linear upper bound $2n-m+1$ is not a consequence of this proposition). Indeed, let $C_m$ be the cycle of the $m$-wheel, and let $v$ be its center. Let $P$ be a shortest path from $C_m$ to $v$, and let $b$ be the first vertex of $P$. Let $a$ be the vertex preceding $b$ in $C_m$, and let $c$ be the vertex succeeding $b$ in $C_m$ (by the minimality of $P$, $a$ and $c$ are not in $P$). Let $C_1$ be the cycle obtained from $P$ and the arc $(v,b)$, and let $C_2$ be the cycle obtained from $P$ and the arcs $(v,a)$ and $(a,b)$. Let $\ell_1$ be the length of $C_1$ and $\ell_2$ be the length of $C_2$. Then $H=G\setminus (b,c)$ is strong and contains both $C_1$ and $C_2$. Since $\ell_2=\ell_1+1$, $H$ is primitive, and thus, by Proposition~\ref{pro:primitive}, $G$ admits a boolean nilpotent function of class at most $n^2-2n+3$. However, the second point of Theorem~\ref{thm:loop} shows that this quadratic upper bound is far from being optimal. Note that if $m=1$, that is if $G$ has a loop, then the fact that $G$ admits a boolean nilpotent function does not result directly from Proposition 7, since, for instance, $C_{1,n}$ does not satisfy the condition of this proposition. 
\end{remark}

The proof of Theorem~\ref{thm:loop} needs an additional definition and a lemma.

\begin{definition}
Given a digraph $G$, and two arcs $(a,b)$ and $(v,w)$ in it, we say that $(a,b)$ is a {\em good arc} for $(v,w)$ if: 
\begin{itemize}
\item 
$b\neq w$, 
\item
for all vertex $u\neq w$, $G\setminus(a,b)$ has a path from $u$ to $v$, 
\item
for all vertex $u$, $G\setminus(a,b)$ has a path from $w$ to $u$, or $G$ has a path from $w$ to $u$ containing $(a,b)$ and such that every vertex in the subpath from $b$ to $u$ is of in-degree one in $G$. 
\end{itemize}
\end{definition}

\begin{lemma}\label{lem:goodarc}
Every arc of a strong digraph with at least two vertices has a good arc. 
\end{lemma}

\begin{proof}
\setcounter{claimcounter}{0}
Suppose, for a contradiction, that $G$ is a smallest counter example with respect to the number $n$ of vertices and then with respect to the number $m$ of arcs. It is straightforward to show that $n\geq 3$. Let $(v,w)$ be an arc of $G$ without good arc. 

\begin{claim}
$G\setminus(a,b)$ is not strong for every arc $(a,b)\neq (v,w)$.
\end{claim}

\begin{subproof}
Indeed, if $G$ has an arc $(a,b)$ with $b\neq w$ such that $G\setminus (a,b)$ is strong, then $(a,b)$ is obviously a good arc for $(v,w)$ in $G$, a contradiction. Furthermore, if $G$ has an arc $(u,w)$ with $u\neq v$ such that $G\setminus (u,w)$ is strong, then $G\setminus (u,w)$ has a good arc $(a,b)$ for $(v,w)$ (since $G\setminus (u,w)$ is not a counter-example), and it is straightforward to show that $(a,b)$ is still a good arc for $(v,w)$ in $G$. This proves the claim.
\end{subproof}

\begin{claim}
$v\neq w$. 
\end{claim}

\begin{subproof}
Suppose that $v=w$. By the first claim, $G'=G\setminus (v,v)$ is a minimal strong digraph and following \cite{GLM12}, $G'$ contains two linear vertices. Thus $G'$ has a linear vertex $c\neq v$. Let $a$ be its unique in-neighbor and let $b$ be its unique out-neighbor. Let $G''$ be the digraph obtained from $G'$ by removing $c$ and adding $(a,b)$. Since $G''$ is strong and is not a counter-example, $G''$ has a good arc $(a',b')$ for $(v,v)$. If $(a',b')\neq (a,b)$ then it straightforward to show that $(a',b')$ is still a good arc for $(v,v)$ in $G$, a contradiction. Otherwise $(a',b')=(a,b)$ and then it is clear that $(a,c)$ is a good arc for $(v,v)$ in $G$, a contradiction. This proves the claim. 
\end{subproof}

\begin{claim}
$v$ is of in-degree at least two. 
\end{claim}

\begin{subproof}
Suppose that $v$ is of in-degree one in $G$. Let $t$ be its unique in-neighbor. Suppose that $t=w$, and let $G'$ be obtained from $G$ by removing $v$ and adding a loop on $w$. Since $G'$ is strong and is not a counter example, $G'$ has a good arc $(a,b)$ for $(w,w)$ and it is easy to see that $(a,b)$ is a good arc for $(v,w)$ in $G$, a contradiction. So suppose that $t\neq w$, and let $G'$ be the strong graph obtained by contracting the arc $(t,v)$ into a single vertex $t'$ ($G'$ has no loop on $t'$ and no multiple arcs). If $t,v\in G(u)$ for some vertex $u$ then $G\setminus (t,u)$ is strong, a contradiction. Thus we have 
\begin{equation}\label{eq:good}
|G'(u)|=|G(u)|~\text{ for all }u\neq t,v,
\qquad
|G'(t')|=|G(t)|
\qquad\text{and}\qquad
\quad
|G(v)|=1.
\end{equation}
Since $G'$ is strong and is not a counter example, $G'$ has a good arc $(a,b)$ for $(t',w)$, and using (\ref{eq:good}) it is  straightforward to show that: 
if $t'\neq a,b$ then $(a,b)$ is a good arc for $(v,w)$ in $G$; if $t'=a$ then either $(t,b)$ or $(v,b)$ is a good arc of $(v,w)$ in $G$; and if $t'=b$ then $(a,t)$ is a good arc of $(v,w)$ in $G$. Thus is every case we have a contradiction. This proves the claim.
\end{subproof}

\begin{claim}
$G$ has a linear vertex $\ell\neq w$. 
\end{claim}

\begin{subproof}
If $G\setminus (v,w)$ is not strong, then by the first claim $G$ is a minimal strong digraph, thus $G$ has two linear vertices and the claim follows. Otherwise $G\setminus (v,w)$ is a strong minimal digraph, thus it has two linear vertices $\ell_1$ and $\ell_2$. Since $v\neq w$, the in-degree of $v$ in $G\setminus (v,w)$ is at least two, thus $v\neq\ell_1,\ell_2$. If $\ell_1\neq v,w$ then $\ell_1$ is a linear vertex of $G$ and otherwise $\ell_2\neq v,w$ thus $\ell_2$ is a linear vertex of $G$. 
\end{subproof}

We are now in position to obtain the final contradiction. Let $\ell\neq w$ be a linear vertex of $G$. Let $\ell^-$ be the in-neighbor of $\ell$, and let $\ell^+$ be its out-neighbor. Let $G'$ be the graph obtained by removing $(\ell^-,\ell)$ and $(\ell,\ell^+)$ and by adding $(\ell^-,\ell^+)$. Since $G'$ is strong and is not a counter-example, $G'$ has a good arc $(a,b)$ for $(v,w)$. It is then straightforward to show that: if $(a,b)\neq(\ell^-,\ell^+)$ then $(a,b)$ is a good arc for $(v,w)$ in $G$; otherwise, both $(a,b)=(\ell^-,\ell^+)$ and $(\ell^-,\ell)$ are good arcs for $(v,w)$ in $G$. Thus in every case we obtain a contradiction.
\end{proof}

Given a digraph $G$, we denote by $G_{ab}$ the signed digraph obtained from $G$ by adding a negative sign to $(a,b)$ and a positive sign to the other arcs. We call {\em $G_{ab}$-and-net} the $G_{ab}$-function $f$ defined~by: 
\[
f_b(x)=\overline{x_a}\land\bigwedge_{j\in G(b)\setminus \{a\}}x_j 
\quad
\text{and}
\quad
f_i(x)=\bigwedge_{j\in G(i)}x_j\text{ for all }i\neq b. 
\]
We are now in position to prove Theorems~\ref{thm:loop}. Actually, under the conditions of the statements, there exists an arc $(a,b)$ such that {\em $G_{ab}$-and-net} is a nilpotent function with the desired properties. For both points, we use Lemma~\ref{lem:goodarc} to find the right arc $(a,b)$.  

\begin{proof}[Proof of Theorem~\ref{thm:loop}]
Suppose that $G$ has vertex set $[n]$. We begin with the second point, which is more easy to prove. So suppose that $G$ has a loop on $v$. Then by Lemma~\ref{lem:goodarc}, $G$ has a good arc $(a,b)$ for $(v,v)$. Thus: for all vertex $i\in[n]$, $G\setminus (a,b)$ has a path $Q_i$ from $i$ to $v$ of length $q_i$; and for all vertex $i\in[n]$, $G$ has a path $P_i$ from $v$ to $i$ of length $p_i$ such that either $(a,b)$ is not an arc of $P_i$ or $(a,b)$ is an arc of $P_i$ and every vertex in the subpath from $b$ to $i$ is of in-degree one in~$G$. We set $\gamma_i=1$ if $(a,b)$ is an arc of $P_i$ and $\gamma_i=0$ otherwise. Consider the $G_{ab}$-and-net $f$. 

\medskip
Suppose that $x_v=0$ and let us prove, by induction on $p_i$, that $f^{p_i+k}_i(x)=\gamma_i$ for all $k\geq 0$. If $p_i=0$ then $i=v$. Since $x_v=0$ and $v$ has a loop, we have $f^k_v(x)=0=\gamma_v$ for all $k\geq 0$. Suppose that $p_i>0$ and let $j$ be the vertex preceding $i$ in $P_i$. By induction, $f^{p_j-1+k}_j(x)=\gamma_j$ for all $k\geq 0$. If $\gamma_j=0$ and $(j,i)\neq (a,b)$ we have $f^{p_i+k}_i(x)=0=\gamma_i$. If $\gamma_j=0$ and $(j,i)=(a,b)$ then by the choice of $P_i$ we have $G(i)=\{j\}$ thus 
\[
f^{p_i+k}_i(x)=\overline{f^{p_i-1+k}_j(x)}=\overline{\gamma_j}=1=\gamma_i.
\]
Finally, if $\gamma_j=1$ then $j\neq a$ and $G(i)=\{j\}$ thus 
\[
f^{p_i+k}_i(x)=f^{p_i-1+k}_j(x)=\gamma_j=1=\gamma_i.
\] 
This completes the induction step. Since $\max p_i\leq n-1$ we have 
\[
x_v=0~\Rightarrow f^{n-1}(x)=\gamma=f(\gamma)
\]
Thus if $x_i=0$ for some $i\in [n]$, then $f^{q_i}_v(x)=0$ and we deduce that $f^{q_i+n-1}(x)=\gamma$, and since $q_i<n$ we obtain $f^{2n-2}(x)=\gamma=f(\gamma)$. Otherwise $x=1$, thus $f_b(x)=0$ and we get $f^{1+q_b+n-1}(x)=\gamma$ and since $q_b<n$ we obtain $f^{2n-1}(x)=\gamma=f(\gamma)$. This proves the second point of the theorem.  

\medskip
\setcounter{claimcounter}{0}
We now prove the first point by following a similar scheme. Suppose that $G$ has a $m$-wheel $W_m$. Let $C$ be the cycle of length $m$ contained in this wheel, and let $v$ be its center. Suppose that $[m]$ is the vertex set of $C$. Let $G'$ be the digraph obtained from $G$ by contracting the cycle $C$ into a single vertex~$c$. Let $(a',b')$ be a good arc for $(v,c)$ in $G'$. We define the arc $(a,b)$ of $G$ as follows: if $c\neq a'$ then $(a,b)=(a',b')$, and if $c=a'$ then $(a,b)$ if any arc of $G$ from $C$ to $b'$. By construction: 
\begin{itemize}
\item
$m<b\leq n$, 
\item
for all vertex $m<i\leq n$, $G\setminus (a,b)$ has a path $Q_i$ from $i$ to $v$ of length $q_i<n$, 
\item
for all $1\leq i\leq n$, $G$ has a shortest path $P_i$ from $C$ to $i$ of length $p_i< n-m+1$ such that either $(a,b)$ is not an arc of $P_i$ or $(a,b)$ is an arc of $P_i$ and every vertex in the subpath from $b$ to $i$ is of in-degree one in~$G$.  
\end{itemize}
We set $\gamma_i=1$ if $(a,b)$ is an arc of $P_i$ and $\gamma_i=0$ otherwise. Consider the $G_{ab}$-and-net $f$. 

\begin{claim}
If $x_i=0$ for all $i\in[m]$, then $f^{n-m}(x)=\gamma=f(\gamma)$. 
\end{claim}

\begin{subproof}
Suppose that $x_i=0$ for all $i\in[m]$. Since $\max p_i\leq n-m$, it is sufficient to prove by induction on $p_i$ that $f^{p_i+k}_i(x)=\gamma_i$ for all $i\in[n]$ and $k\geq 0$. If $p_i=0$ then $i$ is a vertex of $C$, thus there is an arc $(j,i)\neq (a,b)$ in $C$, and thus $f_i(x)=0$; for the same reason we have $f_j(x)=0$ and consequently, $f^k_i(x)=0$ for all $k\geq 0$. Now, suppose that $p_i>0$, and let $j$ be the vertex preceding $i$ in $P_i$. By induction, $f^{p_i-1+k}_j(x)=\gamma_j$ for all $k\geq 0$. If $\gamma_j=0$ and $(j,i)\neq (a,b)$ we deduce that $f^{p_j+k}_j(x)=0=\gamma_i$. If $\gamma_j=0$ and $(w,u)=(a,b)$ then by the choice of $P_i$ we have $G(i)=\{j\}$ thus 
\[
f^{p_i+k}_i(x)=\overline{f^{p_i-1+k}_j(x)}=1-\gamma_j=1=\gamma_i.
\]
Finally, if $\gamma_j=1$ then $j\neq a$ and $G(i)=\{j\}$ thus 
\[
f^{p_i+k}_i(x)=f^{p_i-1+k}_j(x)=\gamma_j=1=\gamma_i.
\]
This completes the induction step.
\end{subproof}

\begin{claim}
If $x_i=1$ for some $i\in [m]$, then $f^{q}_v(x)=0$ for some $q\leq n$. 
\end{claim}

\begin{subproof}
Suppose that $x_i=0$ for some $m<i\leq n$. Since $G\setminus (a,b)$ has a path from $i$ to $v$ of length $q_i<n$ we deduce that $f^{q_i}_v(x)=0$. So suppose that $x_i=1$ for all $m<i\leq n$. If $x_a=1$ then $f_b(x)=0$ and since $m<b\leq n$ we deduce as previously that $f^{q_b+1}_v(x)=0$. Thus, we suppose that $x_a=0$, and consequently $a\in [m]$. If there exists an arc $(i,j)\neq (a,b)$ leaving $C$ ({\em i.e.} with $i\leq m<j$) then $G\setminus (a,b)$ has a path from $a$ to $v$ of length $q_a<n$ and thus $f^{q_a}_v(x)=0$. Finally, suppose that $(a,b)$ is the unique arc leaving $C$; then $Q_b$ does not intersect $C$ and thus $q_b<n-m$. Let $i\in[m]$ with $x_i=1$, and let $j$ be the out-neighbor of $i$ in $C$. Then $x_\ell=1$ for all $\ell\in G(j)$, thus $f_j(x)=1$. Hence, we deduce that $f^r_a(x)=1$ where $r<m$ is the length from $i$ to $a$ in $C$, and thus $f^{r+1}_b(x)=0$, so $f^{r+1+q_b}_v(x)=0$ with $r+1+q_b<n$.
\end{subproof}

Consequently, if $x_i=0$ for all $i\in [m]$ then $f^{n-m}(x)=\gamma=f(\gamma)$ by Claim 1. Otherwise, by Claim 2 we have $f^q_v(x)=0$ for some $q\leq n$. Then, for all $i\in[m]$ we have $v\in G(i)$ and $(v,i)\neq (a,b)$, thus $f^{q+1}_i(x)=0$ and we from Claim 1 that $f^{q+1+n-m}(x)=\gamma=f(\gamma)$. \end{proof}

\begin{remark}
Let $m\geq 2$ and let $W_m$ be the $m$-wheel with center $v$, and let $u$ and $w$ be two consecutive vertices in the cycle of the wheel. Let $G$ be the digraph obtained $W_m$ by adding two vertices $a$ and $b$, and the following arcs: $(a,b)$, $(b,a)$, $(a,u)$, $(w,v)$, $(v,a)$. As in the proof, let $G'$ be obtained from $G$ by contracting $C_m$ into a single vertex $c$. Then $(a,b)$ is the unique good arc for $(v,c)$ in $G'$ and the $G_{ab}$-and-net is a nilpotent function of class $2n-m+1$. This shows that the bound is tight for the class of and-nets with a unique negative arc.   
\end{remark}

\subsection{Nilpotent functions of constant class}\label{sec:constant}

If $G$ is a loop-less digraph on $n$ vertices with minimal in-degree at least two and with a vertex $v$ of out-degree $n-1$, then it is clear that $G$ has a wheel with center $v$. Thus, by Theorems~\ref{thm:loop}, $G$ admits a boolean nilpotent function of class at most $2n-1$. But actually, $G$ admits a boolean nilpotent function of class $3$. 

\begin{proposition} \label{prop:universal}
Every loop-less digraph on $n$ vertices with minimal in-degree at least two and maximal out-degree $n-1$ admits a boolean nilpotent function of class $3$.
\end{proposition}

\begin{proof}
Let $G$ be a digraph on $[n]$ as in the statement. Suppose that vertex $1$ has out-degree $n-1$. Let $f$ be the $G$-function defined by $f_1(x)=\land_{j\in G(1)}\overline{x_j}$ and $f_i(x)=\land_{j\in G(i)}x_j$ for all $1<i\leq n$. Let $\alpha=(1,0,\dots,0)\in\{0,1\}^n$. We have $f(0)=f(\alpha)=\alpha$, and we prove that $f^3(x)=\alpha$ for every $x$ below. First, if $x_1=0$ then $f(x)\in\{0,\alpha\}$ thus $f^2(x)=\alpha$. Second, if $x_1=1$ and $x\neq\alpha$, then $f_1(x)=0$. This brings us back to the first case and we obtain $f^3(x)=\alpha$.
\end{proof}

A digraph is {\em symmetric} if for every arc $(u,v)$, $(v,u)$ is also an arc. We see (undirected) graphs as loop-less symmetric digraphs. Thus the complete graph on $n$ vertices, denoted $K_n$, is the loop-less digraph with $n^2-n$ arcs. Below we prove that, excepted $K_2$, every connected graph admits a boolean nilpotent function of class 3. The proof uses the following notations. We denote by $D(v)$ the set of vertices $u$ such that the distance $d(u,v)$ between $u$ and $v$ in $G$ is at most $2$. For $U\subseteq V$ we set $G(U)=\cup_{v\in U} G(v)$.

\begin{theorem} \label{th:undirected}
Let $G$ be a connected graph. If $G=K_2$, then $G$ admits no boolean nilpotent  functions. Otherwise, $G$ admits a boolean nilpotent function of class $3$.
\end{theorem}

\begin{proof}
The claim is clear for $G =K_2$. Moreover, the case $K_n$, $n \geq 3$ is treated in Proposition \ref{prop:universal} above. If $G$ is not a clique, then we construct a set of vertices $I = \{i_1,\ldots,i_p\}$ as follows. Let $S$ be the set of vertices of degree one in $G$. Let $G^1=G$ and $S^1=S$. We begin the construction of $I$ with a vertex $i_1$ in $S^1$ and set $S^2=S^1\setminus D(i_1)$ and $G^2=G^1\setminus D(i_1)$. Then we pick another vertex $i_2$ in $S^2$ and set $S^3=S^2\setminus D(i_2)$ and $G^3=G^2\setminus D(i_2)$. We continue these processes until $S^k=\emptyset$. Then, we pick a vertex $i_k$ in $G^k$, we set $G^{k+1}=G^k\setminus D(i_k)$, and we continue this process until no more vertex can be added. In this way, the first $k-1$ vertices of $I$ are in $S$ (with $k\geq 2$ if and only if $S\neq\emptyset$) and the remaining $p-k+1$ vertices are not in $S$. 

\begin{claim*}
$I$ is a maximal set of vertices such that $d(i,j)\geq 3$ for all distinct $i,j\in I$ and $d(i,I) \le 2$ for all $i\not\in I$. We also have $S \cap G(I) = \emptyset$.
\end{claim*}

\begin{subproof}
The first part of the claim follows from the construction of $I$. For the second part of the claim, suppose that there exists $s \in S \cap G(I)$. Any two vertices in $S$ cannot be adjacent, so there exists $i\in I \setminus S$ adjacent to $s$. But then  for any vertex $j\neq s$, $d(j,s)=d(j,i)+1$, and thus $d(s, I\setminus i) = d(i,I\setminus i)+1 \geq 4$. This is impossible, because then $s$ would have been chosen to be included in $I$. 
\end{subproof}

By the claim above any vertex in $G(i)$ has a neighbor in $D(i)$. We then consider the and-net $f$ with all arcs signed positively, except those received by vertices $i\in I$, which are signed negatively:
\begin{align*}
	f_i(x) &= \bigwedge_{j \in G(i)} \overline{x_j}\quad\forall i\in I\\
	f_i(x) &= \bigwedge_{j \in G(i)} x_j\quad\forall i\not\in I.
\end{align*}
Let us now consider how the function evolves around a vertex $i\in I$. Let $X_i$ be the set of $x\in\{0,1\}^n$ such that $x_i=1$ and $f_j(x)=0$ for all $j\in D(i)$. First, since every vertex in $D(i)$ has a neighbor in $D(i)$, we have $x\in X_i~\Rightarrow~f(x)\in X_i$. Furthermore, 
\begin{enumerate}
\item
if $x_i=0$ and $x_{G(i)}=0$, then $f(x)\in X_i$. 
\item
if $x_i=0$ and $x_{G(i)}\neq 0$, then $f_i(x)=0$ and $f(x)_{G(i)}=0$ and by Case 1, $f^2(x)\in X_i$.
\item
if $x_i=1$ and $x_{G(i)}\neq 0$, then $f_i(x)=0$ and by Cases 2 and 3, $f^3(x)\in X_i$. 
\item	
if $x_i=1$ and $x_{G(i)}=0$ then $f_i(x)=1$ and $f(x)_{D(i) \setminus G(i)}=0$, thus either $f(x)\in X_i$ or $f(x)_{G(i)}\neq 0$ and by Case 2, $f^3(x)\in X_i$. 
\end{enumerate}
Therefore, in any case, $f^3(x)\in X_i$. Since every $j\not\in I$ belongs to $D(i)$ for some $i\in I$, we deduce that for, any $x\in\{0,1\}^n$, we have $f^3_i(x)=1$ for all $i\in I$ and $f^3_j(x)=0$ for all $j\not\in I$. 
\end{proof}

If $G$ is a loop-less digraph, then $\mathring{G}$ denotes the digraph obtained by adding a loop on each vertex. Here are additional families of digraphs that admits boolean nilpotent function of constant class.   

\begin{theorem} \label{th:G_bar}
If $G$ is a loop-less digraph with minimum in-degree at least one, then $\mathring{G}$ admits a boolean nilpotent function of class $4$. Moreover, if $G$ is symmetric or has a vertex of out-degree $n-1$, then $\mathring{G}$ admits a boolean nilpotent function of class $3$.
\end{theorem}

\begin{proof}
We first consider a set of vertices $I$ as follows. Let $S_0$ be a set of vertices containing exactly one vertex per initial strong component of $G$. For any $d \ge 1$, let $S_d$ be the set of vertices $i$ such that the minimum length of a path from $S_0$ to $i$ is exactly $d$. Suppose that every path can be reached from $S$ by a path of length at most $2m+1$, and let $I= S_0 \cup S_2 \cup S_4 \cup\dots\cup  S_{2m}$. Consider then the boolean $\mathring{G}$-function defined by: 
\begin{align*}
	f_i(x) &= \overline{x_i}\land\bigwedge_{j \in G(i)} x_j\quad\forall i\in I,\\
	f_i(x) &= x_i\land\bigwedge_{j \in G(i)} x_j\quad\forall i\notin I.
\end{align*}
Let us study how $f$ evolves around a vertex $j\notin I$. Since $j\in S_{2r+1}$ for some $0\leq r\leq m$, there exists $i\in G(j)\cap S_{2r}$ (thus $i\in I$). If $x_i=0$ then $f_j(x)=0$ and thanks to the positive loop on $j$ we have $f^{1+k}_j(x)=0$ for all $k\geq 0$; and if $x_i=1$ then thanks to the negative loop on $i$ we have $f_i(x)=0$ and thus $f^{2+k}_j(x)=0$ for all $k\geq 0$. Thus we have proved that
\[
\forall 0\leq r\leq m,~\forall k\geq 0,\qquad f^{2+k}(x)_{S_{2r+1}}=0, 
\]
So $f^{3+k}(x)_{S_{2r+2}}=0$ and we deduce that 
\[
f^{3+k}(x)_{V\setminus S_0}=0.
\]
If $G$ is symmetric, then every vertex in $S_0$ has a neighbor in $S_1$, thus we have $f^{3+k}(x)_{S_0}=0$ and so $f^{3+k}(x)=0$. If $G$ has a vertex $i$ of out-degree $n-1$, we can take 
$S_0=\{i\}$ and $S_1=V\setminus i$. Then $f^{2+k}(x)_{S_1}=0$ and since $i$ has a positive in-neighbor, it has an in-neighbor in $S_1$, and as in the undirected case we get $f^{3+k}(x)=0$. Finally, since $G$ is of positive minimal in-degree, each vertex in $S_0$ has an in-neighbor in $V\setminus S_0$, thus $f^{4+k}(x)_{S_0}=0$ so that $f^{4+k}(x)=0$. 
\end{proof}

We are finally interested in digraphs that admit a boolean nilpotent function of class $2$. Although we cannot completely characterize them, we exhibit large classes of examples and provide a necessary condition to admit a boolean nilpotent function of class $2$.

\begin{proposition} \label{lam:bar(Kn)}~
\begin{enumerate}
\item 
Any digraph where the number of vertices in the intersection of the in-neighborhood of $i$ and the out-neighborhood of $j$ is even for all vertices $i$ and $j$ admits a boolean nilpotent function of class 2.
\item 
The complete graph with loops $\mathring{K}_n$ ($n \ge 2$) admits a boolean nilpotent function of class~$2$.
\item 
Conversely, no digraph with a vertex of in-degree one admits a boolean nilpotent function of class~$2$.
\end{enumerate}
\end{proposition}

\begin{proof}

Suppose that $G$ satisfies the condition of the first case, and let $f$ be the boolean $G$-function defined by: for all $i\in[n]$, $f_i(x)=\oplus_{j\in G(i)} x_j$. Let $i\in[n]$ and let us prove that $f^2(x)=\cst=0$. For all $k\in[n]$, let $p_k$ be the number of $j\in G(i)$ such that $k\in G(j)$. We have 
\[
f^2_i(x)=\bigoplus_{j\in G(j)} f_j(x)=\bigoplus_{j\in G(i)}\bigoplus_{k\in G(j)} x_k=\sum_{k\in [n]} p_kx_k \mod 2.
\]
Since $p_k$ is the size of the intersection between the out-neighborhood of $k$ and the in-neighborhood of $i$, $p_k$ is even and we deduce that $f^2_i(x)=0$.

\medskip
For $\mathring{K}_n$, the case is easily proved for $n=2$; let us then assume $n \ge 3$. Consider the boolean $\mathring{K}$-function defined by: for all $i\in[n]$, $f_i(x) =\overline{x_i}\land\bigwedge_{j \ne i}x_j$. We have $f(0)=0$, and if $x=1$ or if $x$ has at least two zeros, then $f(x)=0$. Finally, if $x$ has exactly one zero, say $x_i=0$, then $f_j(x)=0$ for all $j\neq i$. Thus $f(x)$ has at least two zeroes and by the preceding case $f^2(x)=0$.

\medskip
Finally, suppose that a digraph $G$ has a vertex $i$ with a unique in-neighbor $j$ (we may have $i=j$). Then for any boolean $G$-function $f$, we have $f_i(x)=x_j\oplus\epsilon$ with $\epsilon \in \{0,1\}$. So if $f_j(x)=0$ and $f_j(y)=1$ then $f^2_i(x)=\epsilon$ while $f^2_i(y)=1\oplus\epsilon$, thus $f$ cannot be a nilpotent function of class 2. 
\end{proof}

\section{Conclusion}\label{sec:conclusion}

We have shown that, in the non-boolean case, every signed digraph $G$ admits a very simple dynamics, that is, a dynamics that converges toward a unique fixed point in $k$ steps, with $k\leq \lfloor \log_2 n\rfloor +2$. Such a dynamics is described by a so called nilpotent function of class $k$. In the boolean case, such a function does not necessarily exist, even if we do not take into account the signs of $G$. This leads us to provide, in the unsigned case, some sufficient conditions for the existence of a boolean nilpotent function. All the results are summarized in Table~\ref{table}.

\medskip
Concerning future works, it could be interesting to establish the complexity of deciding if a digraph $G$ admits a boolean nilpotent function. Besides, in this paper, we have focused on systems that converge according to the so called {\em parallel} update schedule, where all components are updated synchronously at each step. It could be interesting to complete this study by considering other update schedules, such as the sequential or block-sequential ones. Finally, it could be interesting to establish a general upper-bound on the minimal convergence time. To be more precise, let $k(G)$ be the minimal class of a boolean nilpotent function on $G$, with the convention that $k(G)=0$ if $G$ admits no boolean nilpotent functions. What is the order of magnitude of $k(G)$ according to $n$? Is it linear with $n$? 

\paragraph{Acknowledgments} We thank an anonymous referee for pointing out works on the loop numbers and suggesting Proposition~\ref{pro:primitive}.

\begin{table}
\begin{tabular}{|l|c|c|c|}
\hline
\begin{tabular}{p{4cm}}
{\bf Classes of digraphs}
\end{tabular} 
&
{\bf Alphabet size} &
{\bf Nilpotent function} &
{\bf Upper bound}
\\\hline\hline
\begin{tabular}{p{4.5cm}}
Signed digraphs
\end{tabular} & $\geq 4$ & yes & $2$ \\\hline
\begin{tabular}{p{4.5cm}}
Signed digraphs
\end{tabular} & $3$ & yes & $\lfloor\log_2n\rfloor +2$ \\\hline
\begin{tabular}{p{4.5cm}}
Signed digraphs with all arcs labelled positively or negatively 
\end{tabular} & $3$ & yes & $2$ \\\hline\hline
\begin{tabular}{p{4.5cm}}
Strong signed digraphs in which all cycles have the same signs
\end{tabular} & $2$ & no &\\\hline
\begin{tabular}{p{4.5cm}}
Cycle $C_n$ 
\end{tabular} & $2$ & no &\\\hline
\begin{tabular}{p{4.5cm}}
Double cycle $C_{\ell,r}$ such that $\min(\ell,r)$ does not divide $\max(\ell,r)$
\end{tabular} & $2$ & no &  \\\hline
\begin{tabular}{p{4.5cm}}
Double cycle $C_{\ell,r}$ such that $\min(\ell,r)$ divides $\max(\ell,r)$
\end{tabular} & $2$ & yes & $2n-1$ \\\hline
\begin{tabular}{p{4.5cm}}
Digraphs containing a primitive spanning strict subgraph
\end{tabular} & $2$ & yes & $n^2-2n+3$ \\\hline
\begin{tabular}{p{4.5cm}}
Strong digraphs containing a loop or a wheel
\end{tabular} & $2$ & yes & $2n-1$ \\\hline
\begin{tabular}{p{4.5cm}}
Loop-less digraphs with minimal in-degree $\geq 2$ and maximal out-degree $n-1$
\end{tabular} & $2$ & yes & $3$ \\\hline
\begin{tabular}{p{4.5cm}}
Loop-less connected symmetric digraphs with $n\geq 3$
\end{tabular} & $2$ & yes & $3$ \\\hline
\begin{tabular}{p{4.5cm}}
Digraphs with a loop on each vertex and with minimal in-degree $\geq 2$
\end{tabular} & $2$ & yes & $4$ \\\hline
\begin{tabular}{p{4.5cm}}
Symmetric digraphs with a loop on each vertex
\end{tabular} & $2$ & yes & $3$ \\\hline
\begin{tabular}{p{4.5cm}}
Digraphs with a loop on each vertex and maximal out-degree $n$
\end{tabular} & $2$ & yes & $3$ \\\hline
\begin{tabular}{p{4.5cm}}
Complete digraphs with a loop on each vertex
\end{tabular} & $2$ & yes & $2$ \\\hline
\end{tabular}
\caption{\label{table} Classes of digraphs that admit or not nilpotent functions, for a given size of alphabet. In the positive case, an upper bound on the class of the nilpotent functions is given, according to the number $n$ of vertices.}
\end{table}

\end{document}